\documentclass{article}
\usepackage{ijcai13}
\usepackage{times}
\usepackage{verbatim}
\usepackage{amssymb}
\usepackage{amsthm}
\usepackage{qtree}

\newtheorem{theorem}{Theorem}
\newtheorem{prop}{Proposition}
\theoremstyle{definition}
\newtheorem{definition}{Definition}

\title{Revenue Maximization via Hiding Item Attributes}
\author{Mingyu Guo\\Department of Computer Science\\University of Liverpool, UK\\Mingyu.Guo@liverpool.ac.uk \And Argyrios Deligkas\\Department of Computer Science\\University of Liverpool, UK\\A.Deligkas@liverpool.ac.uk}

\begin{document}

\maketitle

\begin{abstract} We study probabilistic single-item second-price auctions where
the item is characterized by a set of attributes. The auctioneer knows the actual
instantiation of all the attributes, but he may choose to reveal only a subset
of these attributes to the bidders.  Our model is an abstraction of the
following Ad auction scenario. The website (auctioneer) knows the demographic
information of its impressions, and this information is in terms of a list of
attributes (e.g., age, gender, country of location). The website may hide
certain attributes from its advertisers (bidders) in order to create thicker
market, which may lead to higher revenue. We study how to hide attributes in an
optimal way.  We show that it is NP-hard to solve for the optimal attribute
hiding scheme.  We then derive a polynomial-time solvable upper bound on the
optimal revenue.  Finally, we propose two heuristic-based attribute hiding
schemes.  Experiments show that revenue achieved by these schemes is close to
the upper bound.  \end{abstract}

\section{Introduction} One advantage of Internet advertising is that it offers
advertisers the ability to target customers based on various traits such as
demographics. \cite{Even-Dar07:Sponsored} showed that, for sponsored search of
a given keyword, instead of running a single auction for the keyword, we can
split the whole auction into many separate auctions based on
visitors/impressions' {\em contexts} (e.g., demographics). For example, if we
know and only know the visitors' locations, then each location defines a
context. In this example scenario, splitting based on context means separate
auction for each location. Splitting based on context increases the
advertisers' welfare.  The explanation is simple: after splitting, advertisers
can tailor their bids to the context.  As a result, advertisers generally only
win (impressions from) visitors that they aim to target, and the payments are
also lower, since advertisers only face competition from those targeting
similar visitors. On the other hand, splitting may reduce the revenue received
by the auctioneer (publisher, e.g., website) due to the {\em thin market
problem}: there may be few competitors for some contexts. Actually, if
for every context, there is only one advertiser interested in it, then the
total revenue is $0$ under the standard second-price auction.

\cite{Ghosh07:Computing} observed that having a single auction for all contexts
and having separate auction for each context are not the only two options.
There are other ways to split based on context, and it may lead to much higher
revenue.  The idea explored in \cite{Ghosh07:Computing} is to {\em cluster} the
contexts into bundles, and run separate auction for each bundle.  For example,
suppose there are three different contexts: Beijing, Chicago, and London
(assuming the only contextual information is the location and visitors are only
from these three cities).  We can have one auction for the bundle Beijing and
Chicago (and a second auction for London only).  The interpretation (due to
\cite{Emek12:Signaling}) is that if a visitor is from Beijing or Chicago, then
the auctioneer informs the advertisers that the impression is from one of these
two cities, {\em but not exactly which}.  When this happens, both advertisers
targeting Beijing and advertisers targeting Chicago will compete in the
auction.  Their bids depend on how much they value impressions from Beijing and
Chicago, respectively. Their bids also depend on the conditional probability
that the impression is from Beijing (or Chicago) given that the impression is
from one of these two cities.

To put it more formally, \cite{Ghosh07:Computing} studied probabilistic
single-item second-price auctions (again, interpretation due to
\cite{Emek12:Signaling}).  In such an auction, there is only one item for sale
under a second-price auction, but the item has different possible {\em
instantiations}. The auctioneer knows the actual instantiation but the bidders
do not. The auctioneer may choose to hide certain information from the bidders
if this increases the revenue.  The probabilistic single-item second-price
auction model is an abstraction of the following Ad auction scenario. We have a
website that sells one advertisement slot. That is, there is only one item --
the only advertisement slot, but the item takes many possible instantiations,
due to the fact that visitors/impressions have different demographic profiles.
The auctioneer knows every visitor's demographic profile, and he may hide
certain information from the advertisers.  As mentioned above,
\cite{Ghosh07:Computing} considered hiding information by {\em clustering}: the
auctioneer tells the bidders that the actual instantiation is among several
instantiations.  \cite{Emek12:Signaling,Miltersen12:Send} studied the exact
same model and went one step further. These two papers studied hiding
information by {\em signaling}: the auctioneer sends out different signals, and
the bidders infer the probability distribution of the actual instantiation,
based on the signal received. It is easy to see that signaling is more general
than clustering.  Interestingly, for full information settings (settings where
the auctioneer knows the bidders' exact valuations), \cite{Ghosh07:Computing}
showed that it is NP-hard to solve for the optimal clustering scheme (optimal
in terms of revenue). On the other hand,
\cite{Emek12:Signaling,Miltersen12:Send} both independently showed that, under
the same full information assumption, it takes only polynomial time to solve
for the optimal signaling scheme. This is mostly due to the fact that 
instantiations are treated as divisible goods
under signaling schemes, 

In this paper, we continue the study of revenue-maximizing probabilistic
single-item second-price auctions.  We observe that in practice, {\em Ad impressions
are categorized based on multiple attributes}.  Given this, we argue that the
most natural way to hide information is by {\em hiding attributes}. For
example, let there be three attributes, each with two possible values:

\begin{itemize}
\item Age: Teenager, Adult
\item Gender: Male, Female
\item Location: US, Non-US
\end{itemize}

Together there are $2^3$ possible instantiations. Under the clustering scheme
studied in \cite{Ghosh07:Computing}, the website is allowed to hide information
by bundling any subset of instantiations.  However, not all bundles are
natural. For example, consider the bundle $\{$(Teenager, Male, US), (Adult,
Female, Non-US)$\}$. By creating this bundle, the website basically 
may tell the advertisers that a visitor is either a teenage US male or an adult
Non-US female. This does not appear natural. The signaling scheme studied in
\cite{Emek12:Signaling,Miltersen12:Send} is even more general than clustering,
so it may also lead to unnatural bundles.

On the other hand, attribute hiding always leads to natural bundles.  For
example, the website may hide the location attribute.  That is, if the actual
instantiation is (Teenager, Male, US), then the website may inform the advertisers
that the visitor is a teenage male. By hiding the location attribute, we
essentially created a bundle (Teenager, Male, ?), which consists of both
(Teenager, Male, US) and (Teenager, Male, Non-US).

Based on the above example, it is easy to see that attribute hiding is
clustering with a particular structure.  It should be noted that this
relationship between attribute hiding and clustering does not mean previous
results on clustering apply to our model.  For example, one of the two main
results from \cite{Ghosh07:Computing} is a constructed clustering scheme that
guarantees one half of the optimal revenue (and one half of the optimal social
welfare). The construction does not apply to our model since it generally 
leads
to unnatural bundles.

In this paper, we first show that it is NP-hard to solve for the optimal
attribute hiding scheme.\footnote{We mentioned earlier that
\cite{Ghosh07:Computing} proved a similar result. The authors showed that it is NP-hard to solve
for the optimal clustering scheme. It should be noted that our NP-hardness
result is not implied by  this earlier result, which relied
on reduction involving unnatural bundles. Actually, our
requirement on bundles being natural greatly adds to the difficulty of the
reduction, and our proof is based on completely new techniques.} We then derive
a polynomial-time solvable upper bound on the optimal revenue.  Finally, we
propose two heuristic-based attribute hiding schemes.  Experiments show that
revenue achieved by these schemes is close to the upper bound.  

Besides the aforementioned related work in the computer science literature, bundling has also been well-studied in the economics literature.
\cite{Palfrey83:Bundling} observed that for small numbers of bidders, a
revenue-maximizing auctioneer may choose to bundle the items, and this makes
bidders universally worse-off. On the other hand, for large numbers of bidders,
the auctioneer may choose to unbundle the items, and this hurts the high-demand
bidders while benefiting the low-demand bidders.  \cite{Chakraborty99:Bundling}
quantitatively analyzed the bundling behavior of the auctioneer. The result is
that under a Vickrey auction, for each pair of objects, there is a unique
critical number.  If there are fewer bidders than this number, the seller
chooses to bundle the items, and vice versa.  \cite{Avery00:Bundling} studied
more sophisticated bundling policy, including bundling with discounts and
probabilistic bundling (the probability of bundling occurring depends on the
bids).

\section{Model Description}

There is a single item for sale characterized by $k$ attributes (attribute $1$
to $k$).  Attribute $i$ has $C_i$ possible values, ranging from $0$ to $C_i-1$. 
$m$ is the total number of possible instantiations.
$m=\prod_{i}C_i$.  In this paper, when we mention polynomial time or
NP-hardness, we mean in terms of $m$.

An instantiation whose $i$-th attribute equals $a_i$ is written as
\[(a_1,a_2,a_3,\ldots,a_k)\]

The space of all possible instantiations $\Omega$ is 
\[\{0,\ldots,C_1-1\}\times\{0,\ldots,C_2-1\}\times\ldots\times\{0,\ldots,C_k-1\}\]

\begin{definition}
A {\em natural bundle} $b$ is an element from the following set of all natural bundles (denoted by $\mathcal{B}$):
\[\{0,\ldots,C_1-1,?\}\times\{0,\ldots,C_2-1,?\}\times\ldots\times\{0,\ldots,C_k-1,?\}\]
\end{definition}

Natural bundles are bundles of instantiations resulting from hiding attributes.
An attribute of a natural bundle either takes a specific value, or is represented by a
question mark, which means that this attribute is hidden.  For example, let
$k=5$, given the instantiation $(a_1,a_2,a_3,a_4,a_5)$, if we hide attributes
$1$ and $3$, then it results in the natural bundle $(?,a_2,?,a_4,a_5)$. This
bundle has size $C_1C_3$.  As another example, every instantiation itself
corresponds to a natural bundle of size $1$ (no attribute hidden).  An
instantiation $\omega$ belongs to a natural bundle $b$ if and only if for every
attribute, either $\omega$ and $b$ share the same attribute value, or the attribute is
hidden for $b$.
Unlike the total number of arbitrary bundles, which equals $2^m$, 
the total number of natural bundles is polynomial in $m$, as shown below:

\[|\mathcal{B}|=\prod_{1\le i\le k}(C_i+1)\le 
\prod_{1\le i\le k}C_i^2=m^2\]

The probabilities' of different instantiations are based on a {\em publicly
known} distribution $\Delta(\Omega)$.  To simplify the presentation, when
discussing bidders' valuations, we factor in the probabilities.  For example,
if bidder $i$ values $\omega$ at $5$ when $\omega$ is the actual instantiation,
and $\omega$ happens with probability $0.1$, then we say bidder $i$'s valuation
for $\omega$ is $0.5$.  

Let $n$ be the number of bidders.  Let $v_i(\omega)$ be bidder $i$'s (expected)
valuation for instantiation $\omega$. Following
\cite{Ghosh07:Computing,Emek12:Signaling,Miltersen12:Send}\footnote{Besides the
full information setting, \cite{Emek12:Signaling} also discussed the more
general Bayesian setting.}, we assume full information: the auctioneer knows
the bidders' true valuations.  Again, following previous models, we only
consider bidders with additive valuations.  That is, bidder $i$'s valuation for
bundle $b$, denoted by $v_i(b)$, equals $\sum_{\omega\in b}v_i(\omega)$.

Following previous models,
the auction is the Vickrey auction. 
We use $2(b)$ to denote the revenue for selling $b$ as a bundle.  
$2(b)$ is the second highest value in $\{v_i(b)|1\le i\le n\}$.

\begin{definition}
An {\em attribute hiding scheme} is a way to cluster the instantiations into natural
bundles. An attribute hiding scheme is characterized by a
set of bundles $\{b_1,b_2,\ldots,b_t\}$, satisfying
\begin{itemize}
\item All bundles are natural: $b_i\in \mathcal{B}$ for $1\le i\le t$
\item The bundles are disjoint\footnote{If under
an attribute hiding scheme, two different natural bundles share one common instantiation, then for this instantiation,
it is not clear which attributes we should hide.}: for every pair of $b_i$ and $b_j$, there exists an attribute, so that for this attribute, $b_i$ and $b_j$ take different values (neither is $?$).
\end{itemize}
Under the attribute hiding scheme 
$\{b_1,b_2,\ldots,b_t\}$, instantiations covered by $b_i$ will have their attributes hidden
to match $b_i$. Essentially, instantiations in $b_i$ are sold in a bundle.
Instantiations not covered by any $b_i$ are sold without hiding attributes (sold separately as natural bundles of size $1$).
\end{definition}


Under attribute hiding scheme
$\{b_1,b_2,\ldots,b_t\}$, the revenue of the auctioneer equals
\[\sum_{1\le i\le t}2(b_i)+\sum_{\omega\in \Omega-\cup_{1\le i\le t}b_i}2(\omega)\]

We introduce another function $r$. For $b\in \mathcal{B}$, $r(b)$ represents
the extra revenue obtained by selling $b$ as a bundle, rather than selling 
instantiations in $b$ separately.
We have
\[r(b)=2(b)-\sum_{\omega\in b}2(\omega)\]

The revenue of the auctioneer can then be rewritten as
\[\sum_{1\le i\le t}r(b_i)+\sum_{\omega\in \Omega}2(\omega)\]

The second term of the above expression does not depend on the attribute hiding
scheme. Therefore, the problem of designing optimal attribute hiding scheme is
equivalent to the problem of searching for a set of disjoint natural
bundles $\{b_1,b_2,\ldots,b_t\}$, so that $\sum_{1\le i\le t}r(b_i)$ is
maximized.

\section{Hardness Result}

Previously, \cite{Ghosh07:Computing} showed that it is NP-hard to solve for the
optimal clustering scheme. The proof was by reduction from {\em 3-partition}:
given $3z$ integers, determine whether it is possible to partition them into
$z$ groups with equal sums.  
In this section, we prove a similar result. We show that it is also NP-hard to
solve for the optimal attribute hiding scheme. Our proof is by reduction from
{\em monotone one-in-three 3SAT}~\cite{Schaefer78:The}.  Monotone one-in-three
3SAT is a variant of 3SAT. Monotone means that the literals are just variables,
never negations. One-in-three means that the determination problem is to see
whether there is an assignment so that for each clause, exactly one literal is
true. We emphasize again that our result is not implied by the hardness result from \cite{Ghosh07:Computing}.

\begin{theorem}
It is NP-hard to solve for the optimal attribute hiding scheme.
\end{theorem}

\begin{proof} Let us consider the following monotone one-in-three 3SAT instance with $D$
clauses:

\[(x_{f(1)}\vee x_{f(2)}\vee x_{f(3)})\land
(x_{f(4)}\vee x_{f(5)}\vee x_{f(6)})\land\ldots\]
\[\ldots\land(x_{f(3D-2)}\vee x_{f(3D-1)}\vee x_{f(3D)})\]

There are $3D$ literals, and they are from a list of $E$ variables ($x_1$ to
$x_E$, $f$'s range is between $1$ and $E$).
According to \cite{Schaefer78:The}, it is NP-complete to determine whether
there exists an assignment of the $x_i$, so that the 3SAT instance is true, and
for each clause, there is exactly one true literal.

We will construct a probabilistic single-item auction scenario with $m$
possible instantiations and $n$ bidders. Both $m$ and $n$ are polynomial in
$E$.  We will show that for the constructed scenario, if we are able to solve
for the optimal attribute hiding scheme in polynomial time (in $m$), then we
are able to determine the above 3SAT instance in polynomial time (in $E$). This
implies that it is NP-hard to solve for the optimal attribute hiding
scheme.

Our construction is as follows. 
Let the number of attributes $k$ be $\lceil\log_2(D)\rceil+\lceil\log_2(E)\rceil+11$.
All attributes are binary.
The total number of instantiations $m$ is polynomial in $E$ as shown below.
\[m=2^{\lceil\log_2(D)\rceil+\lceil\log_2(E)\rceil+11}
\le 2^{\log_2(D)+\log_2(E)+13}\]
\[=8192DE\le 8192E^4\]

Our proof relies on the following seven families of natural bundles (Family~\ref{eq:f1} to \ref{eq:f7}):
\begin{equation}
\label{eq:f1}
(\underline{e},\underline{d},0,?,?,0,1,0,1,0,1,0,1)
\end{equation}
\begin{equation}
\label{eq:f2}
(\underline{e},\underline{d},?,0,?,0,1,0,1,0,1,0,1)
\end{equation}
\begin{equation}
\label{eq:f3}
(\underline{e},\underline{d},?,?,0,0,1,0,1,0,1,0,1)
\end{equation}
\begin{equation}
\label{eq:f4}
(\underline{e},\underline{?},0,0,0,?,?,0,1,0,1,0,1)
\end{equation}
\begin{equation}
\label{eq:f5}
(\underline{?},\underline{d},1,?,?,0,1,?,?,0,1,0,1)
\end{equation}
\begin{equation}
\label{eq:f6}
(\underline{?},\underline{d},?,1,?,0,1,0,1,?,?,0,1)
\end{equation}
\begin{equation}
\label{eq:f7}
(\underline{?},\underline{d},?,?,1,0,1,0,1,0,1,?,?)
\end{equation}

In the above, $\underline{e}$ is the binary representation of integer $e$
($1\le e\le E$). The representation's width is $\lceil\log_2(E)\rceil$. 
Similarly, $\underline{d}$ is
the binary representation of integer $d$ ($1\le d\le D$). The representation's
width is $\lceil\log_2(D)\rceil$.  Finally, $\underline{?}$ is $?$ repeated
$\lceil\log_2(E)\rceil$ times (Family~\ref{eq:f5}, \ref{eq:f6}, and \ref{eq:f7}) or 
$\lceil\log_2(D)\rceil$ times (Family~\ref{eq:f4}).

We recall that the problem of designing optimal attribute hiding scheme is
equivalent to the search of disjoint natural bundles
$\{b_1,b_2,\ldots,b_t\}$, so that $\sum_{1\le i\le t}r(b_i)$ is maximized.
Given a natural bundle $b$, $r(b)$ depends on the bidders' valuations.  We will
construct a set of bidders, so that for any natural bundle $b$, $r(b)=0$ by
default. The exceptions are:

\begin{itemize}

\item For $i=1,2,3$, we use $b^i(e,d)$ to represent the natural bundle
characterized by $e$ and $d$ in Family $i$.  $r(b^i(e,d))=1$ if and only if, in
the 3SAT instance, variable $e$ appears in the $i$-th position of clause $d$.

%

\item We use $b^4(e)$ to represent the natural bundle characterized by $e$ in
Family~\ref{eq:f4}.  Let $\#e$ be the number of times variable $e$ appears in
the 3SAT instance. It is without loss of generality to assume $\#e\le D$ (no
literal appears twice in a clause).  
Let $r(b^4(e))=\#e(1-\epsilon)$.  Here, $\epsilon$ is a constant that is less than
$\frac{1}{D}$.  The idea is to make sure that $\#e(1-\epsilon)>\#e-1$.

\item We use $b^5(d)$ to represent the natural bundle characterized by $d$ in
Family~\ref{eq:f5}.  $r(b^5(d))=3$.

\item We use $b^6(d)$ to represent the natural bundle characterized by $d$ in
Family~\ref{eq:f6}.  $r(b^6(d))=3$.

\item We use $b^7(d)$ to represent the natural bundle characterized by $d$ in
Family~\ref{eq:f7}.  $r(b^7(d))=3$.
\end{itemize}

For now, we simply assume that it is possible to construct a polynomial number of
bidders, so that the values of $r(b)$ for different $b$ are indeed as described
above. We will provide the specific construction toward the end.

Let $O$ be an optimal attribute hiding scheme corresponding to the above construction.  If
$r(b)=0$, then it is without loss of generality to assume $b\notin O$.
Therefore, we can ignore bundles not in the above seven families.  Some bundles
from Family~\ref{eq:f1} to \ref{eq:f3} can also be ignored for the same reason.
For presentation purposes, we call the remaining bundles {\em helpful} bundles.
A bundle $b$ is helpful if and only if $r(b)>0$.

Let us consider a fixed variable $e$ ($1\le e\le E$).  $e$ appears $\#e$ times
in the 3SAT instance, so there are exactly $\#e$ pairs of $d$ ($1\le d\le D$)
and $i$ ($1\le i\le 3$), so that $b^i(e,d)$ is helpful.  We use
$b_{e,1},b_{e,2},\ldots,b_{e,\#e}$ to denote these $\#e$ helpful bundles. They
are the only helpful bundles that intersect $b^4(e)$.  If some of these 
bundles are not in $O$, then none of them is in $O$.  The reason is that
$r(b^4(e))=\#e(1-\epsilon)>\#e-1$, so it is better off to add $b^4(e)$ into $O$ (and push out
$b_{e,1}$ to $b_{e,\#e}$ if they are in $O$).
In summary, for $e$ from $1$ to $E$, we must have one of the following two:
\begin{itemize}
\item $b_{e,1},b_{e,2},\ldots,b_{e,\#e}$ are all in $O$. $b^4(e)$ is not in $O$.
\item None of $b_{e,1},b_{e,2},\ldots,b_{e,\#e}$ is in $O$. $b^4(e)$ is in $O$.
\end{itemize}

Let $T$ be the set of $e$ values where $b_{e,1},b_{e,2},\ldots,b_{e,\#e}$ are
all in $O$.  Let $F$ be the set of $e$ values where none of
$b_{e,1},b_{e,2},\ldots,b_{e,\#e}$ is in $O$.  We use $O_{1234}$ to denote the
set of helpful bundles in $O$ that belong to Family~\ref{eq:f1} to \ref{eq:f4}.
We have
\[\sum_{b\in O_{1234}}r(b)=\sum_{e\in T}\#e + \sum_{e\in F}\#e(1-\epsilon)\]
\[=\epsilon\sum_{e\in T}\#e + \sum_{e\in T}\#e(1-\epsilon)+\sum_{e\in F}\#e(1-\epsilon)\]
\[=\epsilon\sum_{e\in T}\#e + (1-\epsilon)3D\]

Let us then consider a fixed variable $d$ ($1\le d\le D$), $b^5(d)$, $b^6(d)$, and $b^7(d)$ pair-wise
intersect. Therefore, in $O$, at most one of them can appear.  Actually, exact
one of them appears.  If none of them appears in $O$, then we can add $b^5(d)$
into $O$, which results in higher revenue.  Let $e_2$ and $e_3$ be the second
and third variables in clause $d$ of the 3SAT instance.  The only helpful
bundles $b^5(d)$ intersects with are $b^2(e_2,d)$ and $b^3(e_3,d)$. By removing
these two from $O$ (if they are in $O$ to start with) and adding $b^5(d)$ into
$O$, the revenue increases.  Therefore, for any $d$ from $1$ to $D$, $O$
contains exactly one of $\{b^5(d),b^6(d),b^7(d)\}$.  
We use $O_{567}$ to denote the
set of helpful bundles in $O$ that belong to Family~\ref{eq:f5} to \ref{eq:f7}.
We have
\[\sum_{b\in O_{567}}r(b)=3D\]
Hence,
\[\sum_{b\in O}r(b)=
\sum_{b\in O_{1234}}r(b)+
\sum_{b\in O_{567}}r(b)=
\epsilon\sum_{e\in T}\#e + (2-\epsilon)3D\]

Let $d$ be a specific value between $1$ and $D$. If $b^5(d)$ belongs to
$O$, then among helpful bundles characterized by $d$ from Family~\ref{eq:f1} to
\ref{eq:f3}, the only helpful bundle that can coexist with $b^5(d)$ is
$b^1(e_1,d)$, where $e_1$ is the first variable in clause $d$ of the 3SAT
instance. In general, no matter which among $\{b^5(d),b^6(d),b^7(d)\}$
appears in $O$, 
among helpful bundles characterized by $d$ from Family~\ref{eq:f1} to
\ref{eq:f3}, there is at most one that can be in $O$.
Therefore, the total number of helpful bundles from Family~\ref{eq:f1} to \ref{eq:f3}
in $O$ is at most $D$.
we have
\[\sum_{e\in T}\#e\le D\]
\[\sum_{b\in O}r(b)=\epsilon\sum_{e\in T}\#e + (2-\epsilon)3D
\le \epsilon D + (2-\epsilon)3D=6D-2D\epsilon
\]

If we are able to solve for the optimal attribute hiding scheme in polynomial
time, then we are also able to determine in polynomial time whether $\sum_{b\in
O}r(b)$ is equal to the upper bound $6D-2D\epsilon$.  If they are equal, then
we have a satisfactory assignment of the 3SAT instance.  For variable $e$,
$b_{e,1}$ to $b_{e,\#e}$ determine whether $e$ is true or not.  If they are all
in $O$, then $e$ is set to be true. Otherwise (if none of them is in $O$), $e$
is set to be false.  When the upper bound is reached, $\sum_{e\in T}\#e=D$,
which implies that under the above assignment, there are exactly $D$ true
literals.  Next, we show that two true literals cannot appear in the same
clause. That is, there is exactly one true literal for each clause under the
assignment, and all clauses are satisfied (there are $D$ true clauses).  Given
$d$, let the variables in clause $d$ be $e_1,e_2,e_3$.  $b^1(e_1,d)$,
$b^2(e_2,d)$, and $b^3(e_3,d)$ are all helpful bundles.  We proved that among
helpful bundles characterized by $d$ from Family~\ref{eq:f1} to \ref{eq:f3},
there is at most one that can be in $O$.  Therefore, only one of
$b^1(e_1,d),b^2(e_2,d),b^3(e_3,d)$ can be in $O$. That is, only one of
$e_1,e_2,e_3$ is set to be true.  

The other direction can be shown similarly. If there is a satisfactory
assignment of the 3SAT instance, then $\sum_{b\in O}r(b)$ should match the
upper bound $6D-2D\epsilon$.  

In conclusion, for the constructed auction setting, it is NP-hard to determine
whether the optimal revenue $\sum_{b\in O}r(b)+\sum_{\omega\in \Omega}2(\omega)$
reaches $6D-2D\epsilon+\sum_{\omega\in\Omega}2(\omega)$.  

Finally, we still need to show that it is possible to construct a polynomial number of bidders, so that the values of $r(b)$ are exactly as described above. Due to space constraint, we present the construction and omit the proof.
\begin{itemize}
\item We construct two bidders who both value every instantiation equally, and the valuation for every instantiation is $L$ ($L>D$). 
\item For every helpful bundle $b$, we construct two new bidders. By default, both bidders value all instantiations in $b$ at $L$ and value all instantiations outside of $b$ at $0$. The exceptions are that one bidder values instantiation $b|_?^0$ at $r(b)+L$ and the other bidder values instantiation $b|_?^1$ at $r(b)+L$. Here, $b|_?^y$ is the instantiation resulting from replacing all $?$ in $b$ by $y$. 
\end{itemize}
%
%
%
\end{proof}

\section{Tree-Structured Attribute Hiding Schemes}

In this section, we study a special family of attribute hiding schemes, which
we call the {\em tree-structured} schemes.  

Let $b$ be a {\em non-unit} natural bundle (bundle of size greater than $1$).  For $b$, at
least one attribute is hidden. 
Let $x$ be one of the hidden attributes of $b$.  We can split
$b$ into $C_x$ disjoint natural bundles by revealing attribute $x$.  The
resulting bundles are $b|_x^0, b|_x^1, \ldots, b|_x^{C_x-1}$.  $b|_x^i$ represents
the natural bundle obtained by replacing the $x$-th attribute of $b$ by $i$.
If $b$ belongs to an attribute hiding scheme $O$, then
after splitting $b$, the new scheme becomes
\[(O-\{b\})\cup\{b|_x^0,b|_x^1,\ldots,b|_x^{C_x-1}\}\]
It is easy to see that the new scheme is still feasible
(the bundles remain disjoint).

Tree-structured attribute hiding schemes are results of {\em recursive
splitting} (revealing attribute) starting from $\{(?,?,\ldots,?)\}$.  At every
step, we either terminate and keep the current scheme, or pick a non-unit
bundle from the current scheme, and split (reveal) one of its attributes. 

\begin{definition}
An attribute hiding scheme $O$ is tree-structured if and only if
it satisfies one of the following:
\begin{itemize}
\item $O=\{(?,?,\ldots,?)\}$: the scheme is simply hiding all attributes and selling all instantiations in a single bundle.
\item There exists a tree-structured attribute hiding scheme $O'$. 
There exists a bundle $b\in O'$ whose $x$-th attribute is hidden.
After splitting $b$ by revealing attribute $x$, the resulting scheme is equivalent to $O$.\footnote{Two schemes are equivalent if they share the same set of non-unit bundles.}
\end{itemize}
\end{definition}

Let us consider an example with three binary attributes.
$\{(?,?,?)\}$ is, by definition, a tree-structured attribute hiding scheme.
Starting from $\{(?,?,?)\}$, if we pick $(?,?,?)$ and reveal its second attribute, then we get
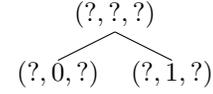
\begin{center}
		\Tree [.$(?,?,?)$ $(?,0,?)$ $(?,1,?)$ ]
\end{center}
The leaves $\{(?,0,?),(?,1,?)\}$ characterize a new
tree-structured attribute hiding scheme.
If we further split the first bundle $(?,0,?)$ based on its third attribute, then we get
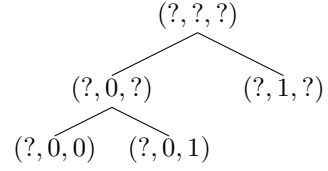
\begin{center}
		\Tree [.$(?,?,?)$ [.$(?,0,?)$ $(?,0,0)$ $(?,0,1)$ ] $(?,1,?)$ ]
\end{center}
Again, the leaves $\{(?,0,0),(?,0,1),(?,1,?)\}$ characterize a new tree-structured attribute hiding scheme.

\begin{prop}
If there are at most two attributes, then all attribute hiding schemes are tree-structured.\footnote{This proposition implies that if there are at most two attributes ($m$ can still be large), then 
we can solve for the optimal attribute hiding scheme in polynomial time, because it must be tree-structured.}
\end{prop}

\begin{prop}
If there are at least three attributes, then there exist attribute hiding schemes that are not tree-structured.
\end{prop}
\begin{proof}
We construct the following natural bundles.
For $i$ from $1$ to $k$, let $b_i$'s $i$-th attribute be hidden, let
$b_i$'s $((i\bmod k)+1)$-th attribute be $1$, and let $b_i$'s all other
attributes be $0$.
\begin{eqnarray}
b_1&=&(?,1,0,0,\ldots,0,0)
\nonumber\\
b_2&=&(0,?,1,0,\ldots,0,0)
\nonumber\\
b_3&=&(0,0,?,1,\ldots,0,0)
\nonumber\\
&\ldots&
\nonumber\\
b_{k-1}&=&(0,0,0,0,\ldots,?,1)
\nonumber\\
b_k&=&(1,0,0,0,\ldots,0,?)
\nonumber
\end{eqnarray}
The $b_i$ are disjoint.
$\{b_1,b_2,\ldots,b_k\}$ is not tree-structured because
starting from $(?,?,\ldots,?)$, if we ever reveal an attribute (e.g., attribute $x$), then
$b_x$ cannot be in the final scheme.
\end{proof}

As we mentioned earlier, tree-structured attribute hiding schemes are results
of recursive splitting starting from the bundle of all instantiations.  At
every step, we either terminate or split a non-unit bundle in some way.
For every natural bundle $b$, let $t(b)$ be the optimal revenue for selling
instantiations in $b$, as a result of making optimal recursive splitting
decisions on $b$.  $t((?,?,\ldots,?))$ is then the optimal revenue of
tree-structured attribute hiding schemes.  Given a bundle, we either sell it as
a whole, or split it in some way as a first step.  Let $h(b)$ be the set of
hidden attributes of $b$.  We have

\[t(b)=\max\{2(b),\max_{x\in h(b)}\sum_{0\le i\le C_x-1}t(b|_x^i)\}\]

If $b$ has size $1$, then $h(b)=\emptyset$. That is, for unit bundles, $t(b)=2(b)$.  Given the
values of $t(b)$ for all $b$ with $|h(b)|=y$, we can then easily compute the values
of $t(b)$ for all $b$ with $|h(b)|=y+1$.  The total number of natural bundles
$|\mathcal{B}|$ is polynomial in $m$.  For every $b$, $t(b)$ is the maximum of
at most $k+1$ values, which is at most $\log_2m+1$. Therefore, the optimal
revenue $t((?,?,\ldots,?))$ can be
computed in polynomial time. The corresponding optimal scheme can be obtained
along the way. 

\section{Upper Bound and Weighted Matching}

Our objective is to find a set of disjoint natural bundles, denoted by $O$, which maximizes
$\sum_{b\in O}r(b)$.  We can model it as an integer program.  We introduce
$|\mathcal{B}|$ binary variables. For $b\in \mathcal{B}$, let $z_b$ be a binary
variable. If $z_b=1$, then it means $b\in O$.
The number of binary variables
$|\mathcal{B}|$ is polynomial in $m$.  The objective is to maximize $\sum_{b\in
\mathcal{B}}z_br(b)$. The constraints are that bundles in $O$ are 
disjoint.  That is, for $b_1,b_2\in \mathcal{B}$, if $b_1$ and $b_2$ intersect,
$z_{b_1}+z_{b_2}\le 1$.  The number of constraints is at most
$|\mathcal{B}|^2$, which is polynomial in $m$.  In summary, the optimal revenue
can be solved for based on an integer program with polynomial numbers of
variables and constraints.  One upper bound can then be solved for in
polynomial time if we consider the linear relaxation (replacing binary
variables by non-integer variables).

Some preprocessing can vastly reduce the number of variables in the above
program.  We first observe that, by definition, $r(b)=0$ for all $b$ with size
$1$. That is, we can safely set $z_b=0$ for all $b$ with size $1$.  We then
observe that, for any natural bundle $b$ with size greater than $1$, if the
following expression is true, then it means that instead of selling $b$ as a
single bundle, we can achieve higher revenue by recursively splitting it, in
which case we can safely set $z_b=0$.

\[2(b)<\max_{x\in h(b)}\sum_{0\le i\le C_x-1}t(b|_x^i)\]

In Section~\ref{sec:experiment}, our simulation shows that when computing the
upper bound, the above observations indeed vastly reduce the number of
variables in the linear program.  For example, for settings with $10$ binary
attributes and $10$ bidders, originally,
there are as many as $(2+1)^{10}=59049$ variables. After preprocessing, there are only
$220.28$ variables on average over repeated simulations.

We then discuss another heuristic for generating attribute hiding schemes with
high revenue.  This heuristic only applies to settings where all attributes are
binary.  If all attributes are binary, then a natural bundle with only one
attribute hidden contains exactly two instantiations.  The heuristic is based
on {\em maximum weighted matching}.  We view all instantiations as vertices. If
two instantiations can be merged into a natural bundle $b$, and $r(b)>0$, then
we create an edge with weight $r(b)$ between them.  Maximum weighted matching
can be solved in polynomial time.  The matching result characterizes the
optimal attribute hiding scheme under the additional constraint that at most
one attribute is hidden.\footnote{In Section~\ref{sec:experiment}, our
simulation shows that there are generally very few natural bundles with at
least two hidden attributes cannot be recursively split to achieve higher
revenue.  This somewhat justifies the heuristic requirement that at most one attribute is
hidden.}

\section{Experiments}
\label{sec:experiment}

In this section, we evaluate the performances of the proposed heuristic-based
attribute hiding schemes. For different values of $k$, $\bar{C}$, and $n$, we
construct problem instances with $k$ attributes, each attribute taking
$\bar{C}$ possible values, and $n$ bidders. The total number of possible
instantiations is then $\bar{C}^k$.  For each instantiation, bidders'
valuations are drawn independently from $U(0,1)$.  For every setup, we repeat
$100$ times and report the averages.

\begin{footnotesize}
\begin{center}
  \begin{tabular}{| c | c | c | c | c | c | c |}
    \hline
Setup & Tree & Match & UB & \#Opt & \#Var & HM\\
    \hline
$k=n=3$ &&&&&&\\
$\bar{C}=2$ & 13.33 & 11.58 & 15.42 & 47 & 5.82 & 1.08\\
    \hline
$k=n=5$ &&&&&&\\
$\bar{C}=2$ & 3.953 & 3.810 & 4.354 & 35 & 15.8 & 1.54\\
    \hline
$k=n=10$ &&&&&&\\
$\bar{C}=2$ & 0.836 & 0.927 & 0.950 & 0 & 220.28 & 4.76\\
    \hline
    \hline
$k=n=3$ &&&&&&\\
$\bar{C}=3$ & 9.251 & NA & 10.58 & 25 & 13.28 & 0.96\\
    \hline
$k=n=5$ &&&&&&\\
$\bar{C}=3$ & 1.767 & NA & 1.976 & 0 & 45.39 & 0.3\\
    \hline
$k=n=8$ &&&&&&\\
$\bar{C}=3$ & 0.296 & NA & 0.361 & 0 & 326.18 & 0.01\\
    \hline
  \end{tabular}
\end{center}
\end{footnotesize}

The table fields are described below:
\begin{itemize}
\item Tree, Match, UB: Comparing to selling all instantiations separately, the
extra revenue in terms of percentage. Tree is short for optimal tree-structured
scheme. Match is short for optimal scheme
based on maximum weighted matching (only applies to $\bar{C}=2$).
UB is short for upper bound on the optimal revenue.
\item \#Opt: Among $100$ repeated simulations, how many times one of 
the heuristic-based schemes reaches the upper bound (therefore guarantees 
optimality\footnote{Even if the heuristic-based schemes do not reach the upper bound, they may still possibly be optimal.}).
\item \#Var: How many variables are in the linear program for computing upper bound.
\item HM: How many natural bundles with at least two hidden attributes cannot
be recursively split to achieve higher revenue.
\end{itemize}

\section{Future Research}

Given the fact that it is NP-hard to solve for the optimal attribute hiding
scheme, one direction of future research is to study whether there are
heuristic-based attribute hiding schemes that guarantee a constant fraction of
the optimal revenue. A similar direction is to see how much revenue we lose by
not allowing unnatural bundles. A preliminary result shows that the optimal revenue
by clustering (allowing unnatural bundles) can be as high as twice the optimal
revenue by hiding attributes. The construction is as follows.  There are $m$
instantiations and $m$ bidders. Bidder $i$ only values instantiation $i$
positively. Let instantiation $1$ be $(0,0,\ldots,0)$ and bidder $1$'s
valuation for it be $\frac{m}{2}$. Let instantiation $m$ be $(1,1,\ldots,1)$
and bidder $m$'s valuation for it be $\frac{m}{2}$. For $1<i<m$, let
bidder $i$'s valuation for instantiation $i$ be $1$. With this setup, the
optimal revenue by clustering is $\frac{2m-2}{2}$. The optimal revenue by
hiding attributes is $\frac{m}{2}$. The ratio $\frac{2m-2}{m}$ approaches $2$
for large $m$. 


\bibliographystyle{named}
\bibliography{../mg}

\end{document}